\documentclass[10pt,conference]{IEEEtran} 
 
\usepackage[utf8]{inputenc}
\usepackage[T1]{fontenc}
\usepackage{amsmath, amssymb, amsthm}
\usepackage{hyperref}
\usepackage{algorithm}
\usepackage{algpseudocode}
\usepackage{booktabs}
\usepackage{titlesec}
\usepackage{stfloats}

\title{\textbf{Proof in a Bottle: Long-Lived Verifiable Secret Sharing via Pre-Quantum Commitment and Immutable Ledger Binding}}

\author{
\IEEEauthorblockN{Markus Jakobsson}
\IEEEauthorblockA{\textit{Artema LABS}\\
New York, USA \\
markus@artemalabs.com}
\and
\IEEEauthorblockN{Keir Finlow-Bates}
\IEEEauthorblockA{\textit{Artema LABS}\\
Eura, Finland \\
keir@artemalabs.com}
}

\date{\today}

\theoremstyle{definition}
\newtheorem{definition}{Definition}[section]
\newtheorem{assumption}{Assumption}[section]

\theoremstyle{plain}
\newtheorem{proposition}{Proposition}[section]

\begin{document}

\maketitle

\begin{abstract}

Traditional secret sharing techniques such as Verifiable Secret sharing (VSS) are vulnerable to quantum attacks by a Cryptographically Relevant Quantum Computer (CRQC) running Shor's algorithm. 

We observe that the binding a VSS needs is required only \emph{at the moment of dealing}, and this binding can be made \emph{before} any CRQC exists. We propose \textbf{Proof in a Bottle} (PiB), which decouples verifiability from long-term binding: standard Pedersen commitments provide zero-knowledge, publicly checkable consistency during a pre-quantum window, while a salted, index-bound hash of the share set, anchored to an immutable public ledger, \emph{preserves} the binding established in that window into the post-quantum era. The guarantee is explicitly a \emph{commit-now, reveal-later} one: it protects today's honest dealings against tomorrow's quantum adversary. 

\end{abstract}

\section{Introduction}
Verifiable Secret Sharing (VSS)~\cite{feldman1987,pedersen1992} lets a dealer distribute shares of a secret so that (1) no fewer than $t$ shareholders learn anything about the secret (\emph{secrecy}), (2) any $t$ or more reconstruct it (\emph{correctness}), and (3) shareholders can check that the shares they hold are mutually consistent with a single degree-$(t{-}1)$ polynomial (\emph{commitment}, or binding). Pedersen VSS realizes this by committing to each polynomial coefficient as $C_j = g^{a_j}h^{r_j}$ and exploiting homomorphism so that each share $(i, y_i)$ is checked against $\prod_j C_j^{i^j}$ without revealing the coefficients.

    We  introduce {\em Proof in a Bottle} (PiB,) which \emph{decouples} zero-knowledge verifiability (Pedersen, pre-quantum) from \emph{long-lived binding} (a salted, index-bound hash anchored to an immutable ledger), formalizing the resulting \emph{commit-now, reveal-later} guarantee.

\subsection{The quantum threat}
The security of Pedersen commitments %splits into two very different properties. 
has two aspects:

\emph{Hiding} is perfect (information-theoretic): a commitment is a uniformly distributed group element, and for every candidate secret there is exactly one randomizer consistent with it, so the commitment reveals nothing, even to an unbounded adversary. \emph{Binding}, by contrast, is computational and rests on DLP.

A CRQC running Shor's algorithm~\cite{shor1994} recovers the trapdoor $x=\log_g h$. This breaks binding only. Concretely, the verification target $P_i=\prod_j C_j^{i^j}$ is a fixed group element with $\log_g P_i = y_i + x R_i$; knowing $x$, an adversary can take \emph{any} value $y_i'$ and solve $R_i' = (\,\log_g P_i - y_i'\,)/x$ so that $g^{y_i'}h^{R_i'}=P_i$. The forged $(y_i',R_i')$ passes the individual Pedersen check, and the forged values need not lie on \emph{any} single degree-$(t{-}1)$ polynomial. Different qualified subsets then reconstruct different secrets: the \textbf{split-brain} failure in which different qualified subsets of shareholders reconstruct different secrets from a single, ostensibly valid dealing.

What Shor does \emph{not} yield is recovery of the true secret: from the commitment alone the adversary learns only the single linear relation $y_i + xR_i$, which every candidate secret satisfies. The failure is thus confined to binding (the commitments become malleable).

We say the commitments become \emph{malleable} when an adversary who has recovered the scheme's trapdoor can open a fixed public commitment to values of its choosing; this is precisely a loss of the \emph{binding} property (equivalently, the ability to \emph{equivocate}), and it leaves \emph{hiding} untouched. 

This distinction matters, because informal descriptions that speak of the adversary ``recovering the secret and randomizers'' overstate the damage: the secret is never exposed, only the guarantee that a unique one was committed.

\subsection{The temporal insight}
The binding a VSS requires is a property of the \emph{act of dealing}: it must be impossible, \emph{once shares are committed}, to later present a different consistent-looking sharing. Crucially, commitments can be timestamped to moments at which DLP is still hard. During such a \emph{pre-quantum window}, Pedersen binding functions normally and sharing is publicly, zero-knowledge verifiable. If, at that moment, the entire share set is frozen against an immutable record, then a \emph{later} CRQC cannot equivocate: the ledger already fixes which shares exist, and the equivocation freedom Shor grants is useless as the target is sealed. This is the mirror image of ``harvest-now, decrypt-later''; we call it \emph{commit-now, reveal-later}. It squarely addresses the realistic threat that commitments made today, under classical assumptions, will be attacked by a CRQC that arrives years later.

\subsection{Contributions}

\section{Related Work}
\textbf{VSS and commitments.} Feldman's VSS~\cite{feldman1987} 
is not malleable, but is vulnerable to  a quantum attacker recovering the committed secret. 
Pedersen's VSS~\cite{pedersen1992}, 
in contrast, protects the committed secret against a quantum attacker, but at the cost of being malleable. 

\textbf{Public verifiability.} Publicly Verifiable Secret Sharing (PVSS)~\cite{stadler1996,schoenmakers1999} lets any observer verify a dealing against public data, but does not protect against quantum attacks.

\textbf{Long-lived integrity.} Anchoring a digest to an append-only record to obtain long-term, tamper-evident timestamps is the classical notarization idea of Haber and Stornetta~\cite{haber1991}, of which a blockchain is a modern, decentralized instantiation. 

PiB is long-term \emph{notarization of a VSS dealing}: one of its novel contributions is to recognize that notarizing a \emph{binding} commitment during the pre-quantum window is sufficient to carry binding across the boundary, without any post-quantum algebraic assumption on the commitment itself.

\textbf{Post-quantum VSS.} Lattice-based VSS (e.g., from the Short Integer Solution problem)~\cite{ajtai1996} offers binding under assumptions believed hard for quantum adversaries, but with large parameters, intricate zero-knowledge machinery, and audit burden. PiB is not a competitor on the same axis: it trades a new hardness assumption for a timing assumption plus an immutable-ledger assumption, and is best suited to settings where dealings can be committed ahead of the quantum horizon.

\section{Preliminaries, Definitions, and Threat Model}

\subsection{Cryptographic primitives}
Let $p$ be a large prime defining $\mathbb{F}_p$, and let $g,h\in\mathbb{Z}_p^{*}$ be public generators with $\log_g h$ unknown to all parties at setup. The Pedersen commitment to $m$ with randomizer $r$ is $C(m;r)=g^{m}h^{r}\pmod p$; it is perfectly hiding and computationally binding under DLP. Let $H:\{0,1\}^{*}\!\to\!\{0,1\}^{N}$ be a cryptographic hash function; we model it as a random oracle where hiding is argued, and rely on its second-preimage resistance for binding. Against a quantum adversary, Grover search gives preimage and second-preimage resistance $\approx 2^{N/2}$, while the BHT algorithm~\cite{bht1998} gives collision resistance only $\approx 2^{N/3}$; our binding argument requires \emph{second-preimage} resistance, so we quote $2^{N/2}$.

\subsection{Security definitions}
\begin{definition}[$(t,n)$-VSS]\label{def:vss}
A dealer shares $s\in\mathbb{F}_p$ among $n$ parties. The scheme is a $(t,n)$-VSS if:
\emph{(Correctness)} any $t$ honestly held shares reconstruct $s$;
\emph{(Secrecy)} any coalition of $<t$ parties, together with all public data, has negligible advantage in distinguishing $s$ from uniform;
\emph{(Binding)} once the dealing is published, no (possibly cheating) party can cause two qualified subsets to reconstruct different values, except with negligible probability.
Here \emph{negligible} means negligible in the hash length $N$, taken over the dated adversary of Section~\ref{sec:threat} (classical before $\tau_Q$, quantum with $O(2^{N/2})$ hash queries after), and conditioned on the ledger-immutability assumption (Assumption~\ref{ass:ledger}); it is not the standard quantification over all PPT adversaries, since our adversary is unbounded in DLP after $\tau_Q$.
\end{definition}

\begin{definition}[Preserved binding across an epoch]\label{def:preserved}
Fix a \emph{quantum horizon} time $\tau_Q$ at which a CRQC first exists. Binding is \emph{preserved} across $\tau_Q$ if the Binding property of Definition~\ref{def:vss} holds against an adversary who is classical before $\tau_Q$ and quantum after, for any dealing whose public commitment is timestamped at some $\tau_0<\tau_Q$.
\end{definition}

\subsection{Threat model}\label{sec:threat}
\begin{itemize}
    \item \textbf{Quantum capability, dated.} The adversary is classical before $\tau_Q$ and gains full Shor/Grover capability after $\tau_Q$. In particular it can solve DLP/ECDLP only after $\tau_Q$.
    \item \textbf{Corruption.} After $\tau_Q$ the adversary may control up to $t-1$ shareholders and may collude with a dealer who wishes to equivocate on a dealing that was \emph{already published} before $\tau_Q$.
    \item \textbf{Ledger.} The adversary cannot rewrite ledger history that has been finalized under post-quantum-secure consensus (see Assumption~\ref{ass:ledger}).
    \item \textbf{Goal.} To induce split-brain: qualified subsets reconstructing different secrets from a pre-$\tau_Q$ dealing.
\end{itemize}

\subsection{Computational Assumptions}
\begin{assumption}[Pre-quantum commitment window]\label{ass:window}
The public commitment of any dealing PiB protects is timestamped at some $\tau_0<\tau_Q$, i.e.\ while DLP is hard, and the dealing is publicly consistency-verifiable during $[\tau_0,\tau_Q)$.
\end{assumption}

\begin{assumption}[Ledger transition and finality]\label{ass:ledger}
Public ledgers in current use authenticate blocks and history with classical signatures and Proof-of-Work/Proof-of-Stake consensus that a CRQC would undermine. We assume that before $\tau_Q$ these ledgers migrate to post-quantum-secure consensus and signatures, and that any anchor PiB relies on is buried under sufficient post-quantum-finalized history \emph{before} $\tau_Q$. Under this assumption the anchor's inclusion and ordering become immutable against the post-$\tau_Q$ adversary.
\end{assumption}

We are explicit that Assumptions~\ref{ass:window}--\ref{ass:ledger} are load-bearing: they replace a post-quantum \emph{algebraic} hardness assumption on the commitment with a \emph{timing plus immutability} assumption on deployment.

\section{The Pre-Quantum Commitment Model}\label{sec:precommit}
PiB is designed for a dealer who wishes to \emph{stockpile} verifiable dealings ahead of the quantum horizon. Concretely, before $\tau_Q$ the dealer generates a large batch of independent $(t,n)$ sharings $\{f^{(m)}\}_{m=1}^{M}$ (for prospective threshold keys, escrowed secrets, key-rotation reserves, or future wallet provisioning) and publishes only the corresponding commitments and anchors (Section~\ref{sec:protocol}). No secret or share need be revealed at publication time. Individual sharings are then \emph{activated} (distributed to shareholders and, when needed, reconstructed) at arbitrary later times, possibly well after $\tau_Q$.

The value of batching is that a single pre-quantum event fixes the binding of the entire reserve. Because each anchor is timestamped at $\tau_0<\tau_Q$ (Assumption~\ref{ass:window}) and finalized under transitioning consensus (Assumption~\ref{ass:ledger}), the binding of \emph{every} sharing in the batch is sealed while DLP still holds, regardless of when it is later opened. This is precisely the \emph{commit-now, reveal-later} posture: the scheme presumes that the honest work of dealing happens today and that the adversary's advantage materializes later, and it removes any benefit from that later advantage with respect to binding.

Two consequences deserve emphasis. First, the window is a hard requirement: a dealing whose anchor is first published \emph{after} $\tau_Q$ enjoys no protection, because a quantum dealer can then anchor an already-inconsistent set (Section~\ref{sec:security}, Proposition~\ref{rem:limit}). Second, the reserve's \emph{secrecy} does not depend on the window at all: Pedersen hiding is perfect and, as shown below, the anchoring is arranged to leak nothing, so stockpiled secrets remain confidential indefinitely, whether opened before or after $\tau_Q$.

\section{The PiB Protocol}\label{sec:protocol}
We describe a single $(t,n)$ sharing; a batch repeats it independently per $m$.

\subsection{Setup}
Let $\mathcal{P}=\{1,\dots,n\}$. The dealer picks a random polynomial $f(x)=a_{t-1}x^{t-1}+\dots+a_1 x + s$ over $\mathbb{F}_p$, with $a_0=s$ the secret.

\subsection{Share generation and double commitment}
\begin{enumerate}
    \item \textbf{Shares.} For each $i\in\mathcal{P}$, compute $y_i=f(i)$.
    \item \textbf{Coefficient commitments.} Pick random $r_j$ and publish $C_j=g^{a_j}h^{r_j}\pmod p$ for $j=0,\dots,t-1$.
    \item \textbf{Share randomizers.} By linearity, $y_i=\sum_{j=0}^{t-1}a_j i^{j}$ and the matching randomizer is $R_i=\sum_{j=0}^{t-1}r_j i^{j}$, so that $g^{y_i}h^{R_i}=\prod_j C_j^{i^{j}}$.
    \item \textbf{Salted, index-bound leaves.} Draw an independent uniform salt $s_i\in\{0,1\}^{\kappa}$ per share and form the leaf
    \[ \ell_i = H(\,i \,\|\, y_i \,\|\, s_i\,). \]
    Binding the index $i$ prevents one shareholder from re-presenting another's value; the secret salt $s_i$ makes the leaf hiding (Proposition~\ref{prop:hiding}).
    \item \textbf{Anchor.} For small $n$, publish the leaf set $\{\ell_i\}_{i=1}^{n}$. For large $n$, publish a single Merkle root
    \[ \mathrm{Root}=H\big(\mathrm{MerkleTree}(\ell_1,\dots,\ell_n)\big). \]
    \item \textbf{Publish (pre-$\tau_Q$).} Post a ledger transaction containing $\mathcal{C}=\{C_0,\dots,C_{t-1}\}$, the anchor (leaf set or $\mathrm{Root}$), and a timestamp $\tau_0$.
\end{enumerate}

\subsection{Distribution}
Privately send to each $i$: the share $(i,y_i)$; the randomizer $R_i$; the salt $s_i$; and (Merkle case) the authentication path $\pi_i$ for $\ell_i$.

\subsection{Verification (zero-knowledge)}
Party $i$ checks, without learning $s$:
\begin{enumerate}
    \item \textbf{Inclusion.} Recompute $\ell_i=H(i\|y_i\|s_i)$ and verify it against the published leaf set, or against $\mathrm{Root}$ via $\pi_i$.
    \item \textbf{Consistency.} Verify
    \[ g^{y_i}h^{R_i} \;\stackrel{?}{=}\; \prod_{j=0}^{t-1} C_j^{\,i^{j}}
       \;=\; g^{\sum_j a_j i^{j}}\,h^{\sum_j r_j i^{j}}. \]
\end{enumerate}
Both checks are zero-knowledge in $s$. Under Assumption~\ref{ass:window}, the consistency check is meaningful (DLP hard) at least once during $[\tau_0,\tau_Q)$; public verifiability means any observer, not only shareholders, may perform it and record success on-ledger if desired.

\subsection{Reconstruction}
Any $\mathcal{S}\subseteq\mathcal{P}$ with $|\mathcal{S}|=t$ recovers
\[ s=f(0)=\sum_{i\in\mathcal{S}} y_i \prod_{j\in\mathcal{S},\,j\neq i}\frac{-x_j}{x_i-x_j}\pmod p, \]
using only shares whose inclusion and consistency checks passed.

\section{Security Analysis}\label{sec:security}

\subsection{Preserved binding}
\begin{proposition}[Binding preservation]\label{prop:binding}
Under Assumptions~\ref{ass:window} and~\ref{ass:ledger}, and assuming $H$ is second-preimage resistant against quantum adversaries ($\approx 2^{N/2}$ work), PiB satisfies preserved binding (Definition~\ref{def:preserved}): except with negligible probability, no post-$\tau_Q$ adversary can cause any qualified subset to reconstruct a value other than the secret $s^{\star}$ fixed at $\tau_0$.
\end{proposition}

\begin{proof}[Proof sketch]
By Assumption~\ref{ass:window}, at $\tau_0<\tau_Q$ DLP is hard, so Pedersen binding holds: the published $\mathcal{C}$ computationally determine unique coefficients $\{a_j\}$, hence a unique polynomial $f$ and shares $y_i^{\star}=f(i)$, and the leaves $\ell_i^{\star}=H(i\|y_i^{\star}\|s_i)$ that any observer can (and does) verify for consistency during the window. The anchor fixes exactly these leaves.

To induce split-brain after $\tau_Q$, the adversary must present, for some $i$ in a qualified subset, a value $y_i'\neq y_i^{\star}$ that passes both checks. The consistency check is now trivially satisfiable: knowing $x$, set $R_i'=(\log_g P_i - y_i')/x$. The inclusion check, however, requires a salt $s_i'$ with $H(i\|y_i'\|s_i')=\ell_i^{\star}$ (leaf-set case) or a path yielding the fixed $\mathrm{Root}$ (Merkle case). The former is a \emph{second preimage} on $H$; the latter additionally requires either a second preimage or an internal Merkle collision. Both are infeasible at $\approx 2^{N/2}$. The remaining avenue, namely replacing the anchor itself, requires rewriting finalized ledger history, excluded by Assumption~\ref{ass:ledger}. Hence every share that passes both checks equals $y_i^{\star}$, and reconstruction yields $s^{\star}$.
\end{proof}

\begin{table*}
    \centering
    \caption{}
    \label{tab:comparison}
    \begin{tabular}{lcc}
        \toprule
        Feature & PiB (proposed) & Lattice-based VSS \\
        \midrule
        Binding basis & 2nd-preimage res.\ + immutable ledger & SIS/LWE (believed PQ-hard) \\
        Binding type & Preserved from pre-quantum window & Native post-quantum \\
        Timing assumption & Required (commit before $\tau_Q$) & None \\
        Ledger assumption & Required (PQ-finalized anchor) & None \\
        Hiding & Perfect (Pedersen); anchor statistical & Scheme-dependent \\
        Verification cost & $O(t)$ modular exponentiations & Matrix ops + lattice ZK \\
        Communication & Small integers + Merkle path & Large keys/ciphertexts \\
        Implementation risk & Low (standard libraries) & High (side channels, complex ZK) \\
        \bottomrule
    \end{tabular}
\end{table*}

\begin{proposition}[No secret recovery]\label{prop:norecovery}
Even after $\tau_Q$, an adversary controlling $<t$ shares learns nothing about $s^{\star}$ from $\mathcal{C}$ and the anchor beyond negligible advantage; in particular Shor does not recover $s^{\star}$.
\end{proposition}
\begin{proof}[Proof sketch]
From $C_0=g^{s}h^{r_0}$ and $x=\log_g h$, the adversary obtains only $\log_g C_0 = s + x r_0$: one linear equation in the two unknowns $s,r_0$, satisfied by every candidate secret. This is Pedersen's perfect hiding, which Shor does not affect. The anchor is handled in Proposition~\ref{prop:hiding}.
\end{proof}

\subsection{Retained secrecy despite public hashing}
The naive scheme (hashing raw shares) breaks threshold secrecy: any party can confirm a guessed share by a single hash evaluation, collapsing Shamir's perfect secrecy at the threshold boundary to a checkable computation. Salting repairs this.

\begin{proposition}[Threshold secrecy]\label{prop:hiding}
Model $H$ as a random oracle. If each salt $s_i$ is uniform on $\{0,1\}^{\kappa}$ and known only to the dealer and party $i$, any coalition $\mathcal{T}$ with $|\mathcal{T}|<t$ (given $\mathcal{C}$, the anchor, all authentication paths it legitimately holds, and its own shares, randomizers, and salts) has advantage at most $O(q\,2^{-\kappa})$ in determining $s^{\star}$, where $q$ is its number of oracle queries. This holds against a post-$\tau_Q$ adversary.
\end{proposition}

\begin{proof}[Proof sketch]
Consider any $i\notin\mathcal{T}$. The coalition's view of $\ell_i=H(i\|y_i\|s_i)$ is, in the random-oracle model, independent of $y_i$ unless the coalition queries the oracle at the exact point $i\|y_i\|s_i$; lacking $s_i$ (uniform, secret), each guess of $y_i$ succeeds with probability $2^{-\kappa}$ per query, giving the stated $O(q\,2^{-\kappa})$ bound. Shor does not help: post-$\tau_Q$ the adversary learns $\log_g P_i = y_i + x R_i$, one equation in two unknowns per missing share, so the $\{y_i\}_{i\notin\mathcal{T}}$ retain a free coordinate exactly as in plain Shamir. The Pedersen commitments contribute nothing further, being perfectly hiding. Hence the missing shares are statistically undetermined and $s^{\star}$ is hidden up to the salt-guessing term. Choosing $\kappa\ge 2N$ keeps this negligible even against Grover.\footnote{Secrecy is thus \emph{statistical}, not perfect: the published leaf hashes admit a vanishing salt-guessing advantage. This is the price of a public, non-interactive anchor and is made negligible by the salt length.}
\end{proof}

\subsection{Ledger immutability and the transition assumption}
Assumption~\ref{ass:ledger} is where PiB's binding ultimately rests, and it is a systems assumption, not an algebraic one. A ledger authenticated only by classical signatures is itself Shor-breakable; PiB therefore requires that the chosen ledger migrate to post-quantum consensus and signatures, and that any relied-upon anchor be buried under post-quantum-finalized depth, \emph{before} $\tau_Q$. Grover's quadratic speedup on Proof-of-Work for chain re-organization is a constant-factor security erosion that must be absorbed by finalization depth or by post-quantum finality gadgets. We regard the timely execution of this migration (already underway in standardization~\cite{nist2024}) as the principal external dependency of the scheme, and note that a dealing loses no \emph{secrecy} if the assumption fails, only the binding.

\begin{proposition}[Scope of the guarantee]\label{rem:limit}
PiB does \emph{not} bind a dealer who is already quantum-capable at commit time. Such a dealer can, at $\tau_0\ge\tau_Q$, anchor a set of mutually inconsistent shares that nonetheless pass every individual Pedersen check (choosing each $R_i$ to fit), baking split-brain into the frozen set. Assumption~\ref{ass:window} excludes this case by fiat, and it is the realistic case: the concern motivating post-quantum migration is that \emph{present} honest dealings will face a \emph{future} CRQC.
\end{proposition}

\subsection{Comparison with lattice-based VSS}
Table~\ref{tab:comparison} contrasts PiB with a post-quantum cryptographic lattice-based approach. Neither is unconditional: lattice VSS rests on SIS/LWE (believed quantum-hard), PiB on second-preimage resistance plus a timing-and-immutability deployment assumption.

\section{Implementation Considerations}
\begin{itemize}
    \item \textbf{Hash size.} Binding rests on second-preimage resistance; against Grover this is $\approx 2^{N/2}$, so $N\ge 256$ (e.g.\ SHA3-256) gives $128$-bit security, and SHA3-512 a margin. We do \emph{not} rely on collision resistance, whose quantum bound (BHT, $\approx 2^{N/3}$) would be weaker.
    \item \textbf{Salt length.} Per Proposition~\ref{prop:hiding}, choose $\kappa\ge 2N$ so the salt-guessing advantage is negligible even under Grover; salts are distributed privately with shares and never published.
    \item \textbf{Ledger selection and timing.} Select a chain with a credible post-quantum consensus roadmap and anchor batches early, leaving ample finalized depth before $\tau_Q$ (Assumption~\ref{ass:ledger}).
    \item \textbf{Batch scale.} For a reserve of $M$ sharings over $n$ parties, a single Merkle root per sharing (or one aggregate root over all batches) keeps on-ledger footprint small; inclusion proofs are $O(\log n)$.
\end{itemize}

\section{Limitations}
The scheme's guarantees are conditional and we restate the conditions plainly. (i) Preserved binding requires the commitment to predate $\tau_Q$ (Assumption~\ref{ass:window}); it says nothing about dealings first anchored in the quantum era. (ii) Binding rests on the ledger genuinely finalizing history under post-quantum consensus before $\tau_Q$ (Assumption~\ref{ass:ledger}); this is a deployment and governance dependency outside the cryptography. (iii) Anchor-induced secrecy is statistical, not perfect, made negligible by salt length. (iv) PiB provides no defense against a dealer who is quantum-capable at commit time (Proposition~\ref{rem:limit}); protecting against such a dealer requires a commitment whose \emph{binding} is itself post-quantum, at which point the Pedersen layer becomes redundant. Within these bounds, PiB is best understood as long-term, publicly verifiable notarization that carries a classically established VSS binding safely across the quantum horizon.

\section{Conclusion}
Post-quantum failure of Pedersen VSS affects \emph{binding}, not \emph{hiding}: Shor enables equivocation but does not reveal the secret. PiB secures binding at the time of dealing and preserves it by anchoring a salted, index-bound hash to an immutable, transitioning ledger, while retaining Pedersen's pre-quantum zero-knowledge verification. The result is a \emph{commit-now, reveal-later} approach that protects today's honest dealings against future quantum adversaries using only classical primitives. Rather than replacing lattice-based cryptography, PiB provides a practical and auditable option for long-term verifiable secret sharing.

For any questions about licensing of this technology \cite{pib-patent}, please contact the organization of the authors.  

\bibliographystyle{IEEEtran} % Uses the standard IEEE numbering style
\bibliography{references}    % Links to references.bib (omit the .bib extension)

\end{document}